\documentclass[11pt]{article}
\usepackage{a4,amsmath,amssymb,amsthm,amscd,cite,graphicx}
\usepackage{rotating}
\usepackage[table]{xcolor} 
\usepackage{mathtools, booktabs, caption, float, multirow, comment} 
\usepackage{fixmath, geometry}
 
\usepackage[hidelinks]{hyperref} 
\usepackage{cleveref}
\usepackage{tikz}
\usetikzlibrary{intersections,calc,matrix,arrows, decorations.markings,shapes}
\usetikzlibrary{arrows.meta}
\usepackage{multirow}
\newtheorem{theorem}{Theorem}[section]
\usepackage{longtable}
\definecolor{mColor1}{rgb}{0.9,0.9,0.9}  
\definecolor{mColor2}{gray}{0.8}
\definecolor{mColor3}{gray}{1.0}

\newtheorem{lemma}[theorem]{Lemma}
\newtheorem{remark}[theorem]{Remark}
\newtheorem{proposition}[theorem]{Proposition}
\newtheorem{definition}[theorem]{Definition}

\newtheorem{fact}[theorem]{Fact}

\newcommand{\be}{\begin{equation}}
\newcommand{\ee}{\end{equation}}
\newcommand{\bea}{\begin{eqnarray}}
\newcommand{\eea}{\end{eqnarray}}

\catcode`@=11 \@addtoreset{equation}{section} \catcode`@=12

\begin{document}

\begin{titlepage}
\begin{flushright}
 (v 1.0) February 2026\\
\end{flushright}
\begin{center}
\textsf{\large S-matrices in the holomorphic modular bootstrap approach}\\[12pt]
Suresh Govindarajan$^{\dagger}$,  Aditya Jain$^{*}$ and  Akhila Sadanandan$^{\ddag}$ \\[3pt]

Email: $^\dagger$suresh@physics.iitm.ac.in, \\ $^*${adityaj2807@gmail.com }
$^\ddag$akhila@physics.iitm.ac.in \\[6pt]
Department of Physics,\\
Indian Institute of Technology Madras\\
Chennai 600036, India \\[6pt]
and \\
Centre for Operator Algebras, Geometry, Matter and Spacetime, \\Indian Institute of Technology Madras, Chennai 600036 India
\end{center}
\begin{center}
\mbox{} \\
Appendix by 
Abhiram Kidambi \\
Max Planck Institute for Mathematics in the Sciences,\\
 Inselstra\ss e 22, 04103 Leipzig, Germany \\
 and \\
  AEI Potsdam, Am Mühlenberg 1, 14476 Potsdam, Germany \\
Email address: kidambi@mis.mpg.de
\end{center}
\begin{abstract}
We numerically determine the S-matrix by using connection formulae in the modular linear differential equation (MLDE) approach to the holomorphic modular bootstrap. We then determine exact formulae using the fact that entries in the $S$-matrix are integer entries in a cyclotomic extension of the field of rational numbers. This provides a method that is intrinsic to the MLDE setup and does not require inputs outside this framework. The method is illustrated with a selection of examples.
\end{abstract}
\end{titlepage}

\section{Introduction}

The key idea in the holomorphic modular bootstrap to rational conformal field theories (RCFTs) is to view the Virasoro characters of the RCFT as independent solutions to a modular linear differential equation (MLDE) of order equal to the number of characters\cite{Mathur:1988na,Mathur:1988gt,Naculich:1988xv}.  A second parameter called the Wronskian index then determines the MLDE up to a certain number of parameters. The goal in the holomorphic modular bootstrap is to determine the values of parameters that lead to solutions that have $q$-series with non-negative integral coefficients -- such solutions are called \textit{admissible}. This is a diophantine problem that is hard to solve without additional inputs and has been used successfully for the theories with three characters or less and small values of Wronskian index\cite{Chandra:2018pjq,Franc:2020,Das:2021uvd,Gowdigere:2023xnm}. 

The set of possible exponents (modulo one) have been obtained for theories with up to five characters\cite{Kaidi:2021ent}. Clearly this determines the $T$-matrix of the RCFT.  This provides additional inputs to the MLDE approach thereby simplifying the search for solutions of the Diophantine equation. However, given an admissible solution, we do not know the $S$-matrix unless the MLDE is of hypergeometric type.  Thus the \textit{modular} data for a given admissible solution is incomplete. In particular, we cannot determine the fusion matrix using the Verlinde formula to decide whether the admissible solution is a potential CFT.

There exist indirect methods that let us determine the $S$-matrix. We list a few of the methods that have appeared in this context. The generalized coset construction of Gaberdiel-Hampapura-Mukhi lets one construct new solutions as duals to known solutions\cite{Gaberdiel:2016zke}. A second method is the use of Hecke operators to obtain new solutions from known ones\cite{Harvey:2018rdc,Duan:2022ltz}.  A third method is to use theory of vector valued modular forms to construct new solutions from old ones\cite{Bantay:2005vk,Gannon:2013jua,Rayhaun:2023pgc,Govindarajan:2025rgh}. A fourth method uses an involutive duality of Bantay-Gannon to relate two different theories\cite{Govindarajan:2025jlq}. All these methods not only enlarge the families of admissible solutions but also provide the $S$-matrix of the new admissible solutions.
This paper provides a method of computing the $S$-matrix that lies within the framework of the MLDE approach. This removes one of the issues with the holomorphic modular bootstrap.

Our approach to determining the S-matrix is to consider the MLDE in the variable $z=\frac{J(\tau)}{1728}$ space. Then, $z=0,1$ and $\infty$ are regular singular points of the MLDE with the monodromy around $z=0$ being the $U$-matrix, $z=1$ being the $S$-matrix and $z=\infty$ being the $T$-matrix.   The Frobenius power series around each singularity only determines the conjugacy class of the monodromy. We switch variables to $w=\frac1z$ so that the $T$-matrix is the monodromy about $w=0$ and $S$ is then the monodromy about $w=1$. We match the two local solutions about $w=0$ and $w=1$ at $w=w_0=\frac12+\epsilon$ with $0< |\epsilon| <\frac12$. (see  \Cref{contours}) The two sets of solutions are convergent in this domain. The change of basis that relates the two solutions is called the \textit{connection matrix} which we denote by the $n\times n$ matrix $C$. The connection matrices are explicitly known only for the generalized hypergeometric functions from the work of Beukers and Heckman\cite{Beukers1989}. Our methods are applicable for the non-hypergeometric MLDEs as well.

The matching is done numerically leading to a numerical estimate for the connection matrix $C$ from which we are able to obtain a numerical estimate for the $S$-matrix. The numerical estimate for the $S$-matrix is then converted to an \textit{exact} one by using the fact that the entries for the normalized $S$-matrix, $\bar{S}=S/S_{00}$ lie in the ring of integers, $\mathbb{Z}[\zeta_{\bar N}]\in \mathbb{Q}(\zeta_{\bar N})$, where $\mathbb{Q}(\zeta_{\bar N})$ is the cyclotomic extension of the rationals by a primitive $\bar N$-th root of unity. Here ${\bar N}$ the order of $\bar{T}=T/T_{00}$.  It is easy to see that ${\bar N}$ is the LCM of the denominators of the conformal weights.

\begin{figure}[ht]
\centering
\begin{tikzpicture}[scale=1.8,
  arrow/.style={-Stealth,thick},
  every node/.style={font=\small}
]

\draw (-3,0)--(3,0);
\draw (0,-2)--(0,2);

\fill (-1.0,0) circle (1.2pt) node[below]{$0$};
\fill ( 1.0,0) circle (1.2pt) node[below]{$1$};
\fill (0,0) circle (1.4pt) node[below right]{$w_0$};

\draw[thick] (-1.0,0) circle (1.0);
\draw[thick] ( 1,0) circle (1.0);

\draw[thick] (0,0) ellipse (2.3 and 1.7);

\draw[arrow] (-1.1,-1)--(-1,-1);
\node at (-1.1,-1.15){$T=M_0$};

\draw[arrow] (1,-1)--(1.1,-1);
\node at (0.88,-1.15){$S=C^{-1}\cdot M_1 \cdot C$};

\draw[arrow] (-0.2,-1.7)--(0.2,-1.7);
\node at (0.2,-1.9){$M_\infty$};


\end{tikzpicture}
\caption{The contours that determine the monodromy data of the MLDE}\label{contours}
\end{figure}
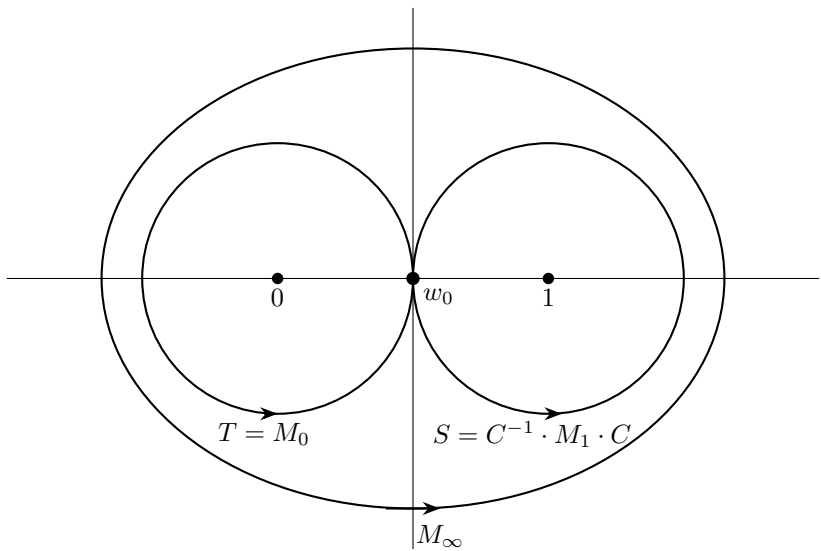

\section{Background}

\subsection{Rational Conformal Field Theories}
The modular invariant torus partition function of an RCFT, with $n$ characters, has the form
\[
Z(\tau,\bar\tau) = \sum_{i,j=0}^{n-1} \bar{\chi}_i(\bar\tau)\ m_{ij}\ \chi_j(\tau)\ ,
\]
where $M=(m_{ij})$ is the multiplicity matrix. The characters have the following 
$q$-series
\begin{equation} \label{exponents}
\chi_i(\tau) = q^{\alpha_i}\ \left(\sum_{m=0}^\infty a_{i,n}\ q^m\right)\ , 
\end{equation}
where $q=\exp(2\pi i\tau)$,  $\alpha_i = -\frac{c}{24}+h_i$ and $a_{0,0}=1$. Further all the $a_{i,n}$ are non-negative integers.

It is useful to form the vector valued modular form $\mathbb{X}(\tau)=(\chi_0,\chi_1,\ldots,\chi_{n-1})^T$ from the characters of the RCFT. The $S$ and $T$ matrices are defined by how they transform under modular transformations:
\begin{align*}
\mathbb{X}(\tau+1) & =  T \cdot \mathbb{X}(\tau)\ ,\\
\mathbb{X}(-\tfrac1\tau) &= S \cdot \mathbb{X}(\tau)\ .
\end{align*}
Modular invariance of the torus partition function requires (i) $c=\bar{c}$, (ii) $h_{i}-\bar{h}_j=0\mod \mathbb{Z}$, when $m_{ij}\neq0$ and $S^\dagger\cdot M \cdot S=M$. For unitary RCFTs, the index $i=0$ is associated with the identity operator which has $h=0$.

Our focus will be mostly on diagonal modular invariants, where $M$ is a diagonal matrix  with entries $(m_0=1,m_1,\ldots,m_{n-1})$ that we will call \textit{multiplicities}. The torus partition function in this case is then $Z=\sum_{i}m_i\, |\chi_i(\tau)|^2$. The number of primaries is given by $\sum_{i=0}^{n-1}m_i$ and can exceed the number of characters.

Let $\phi_i$ represent the $i$-th primary field. The fusion algebra are then given by 
$$\phi_i \times \phi_j =\sum_{k} {N_{ij}}^k\phi_k\ .$$
The  structure constants of the fusional algebra are defined through the Verlinde formula\cite{Verlinde:1988sn}. One has
\[
{N_{ij}}^k =\sum_{a} \frac{S_{ia}S_{ja}S^*_{ak}}{S_{00}}\in \mathbb{Z}_{\geq0}\ .
\]

\subsection{Characters are solutions of an MLDE}

The $n$-characters of an RCFT are solutions to  an $n$-th order modular differential linear equation (MLDE). In the holomorphic modular bootstrap, one begins with an $n$-th order MLDE with fixed Wronskian index $\ell$ and solves a diophantine equation that finds the values of the parameters in the MLDE that lead to solutions with non-negative $q$-series\footnote{The Wronskian index $\ell =\frac{n(n-1)}{2}-6\sum_{i=0}^{n-1}\alpha_i$, where $\alpha_i$ are the exponents given in \Cref{exponents}.}. Such solutions are called \textit{admissible}. Let us denote the solutions by $\widetilde{\mathbb{X}} = (\tilde\chi_0,\tilde\chi_1,\ldots,\tilde\chi_{n-1})^T$. However, there are a couple of issues with the admissible solution even if it is associated with an RCFT.
\begin{enumerate}
\item We do not know the correct normalization for the admissible solutions. Except for the identity character, in all other cases, one has for $i>1$ 
$\chi_i(\tau) = y_i\ \tilde\chi_i(\tau)$, where the $y_i\in \mathbb{Z}_{>0}$ are what we call the \textit{degeneracies}. These are not fixed by the MLDE.
\item We do not immediately know the $S$-matrix of the solutions.
\item We do not know if the RCFT is unitary or not.
\end{enumerate}
In this paper, we will obtain the $S$-matrix within the MLDE approach and this can, on occasion, fix the degeneracies as well.

\subsection{Galois symmetry in RCFT}

The appearance of Galois symmetry in RCFT was first observed by de Boer and Goeree\cite{DeBoer:1990em}.
In an appendix to their paper, they
study the number field, $L$, generated by  the ratios $\lambda_{i}^{(m)} := \frac{S_{im}}{S_{0m}}$ of an RCFT. For fixed $i$, the $\lambda_{i}^{(m)}$ are the eigenvalues of the fusion matrix $N_{i}$ -- this is the fusion coefficient ${N_{ij}}^k$ with  $(j,k)$ considered as indices of the fusion matrix. 
Thus $L$ is a normal extension of $\mathbb{Q}$ and further they prove that the Galois group $\mathrm{Gal}(L/\mathbb{Q})$ is abelian. It follows from the theorem of Kronecker and Weber that $L$ is contained in a cyclotomic extension of $\mathbb{Q}$. However, de Boer and Goeree do not specify the extension. This was done by Ng and Schauenburg\cite{Ng:2007}. One defines the normalized $S$ and $T$ matrices as follows:
\begin{equation}
\bar{T} = \frac{T}{T_{00}}\quad,\quad \bar{S}=\frac{S}{S_{00}}\ .
\end{equation}
Then, the entries of the normalized $S$-matrix, $\bar{S}$ lie in the integer ring $\mathbb{Z}[\zeta_{\bar{N}}]\subset\mathbb{Q}(\zeta_{\bar{N}})$, where $\bar{N}$ is the order of $\bar{T}$, the normalized $T$-matrix\cite{Ng:2023}. \Cref{Kidambi} by Kidambi provides a compact proof of this result.

Another result due to  Coste and Gannon shows that 
all the entries in the $S$-matrix of an RCFT are contained in the field $\mathbb{Q}(\zeta_N)$ where $N$ is the order of the $T$ matrix\cite{Coste:1993af}. Further, the action on the Galois symmetry  on the chiral primaries explains why some entries of the $S$-matrices are equal upto a sign\cite{Gannon:2003de}.

\subsection{Non-unitary RCFTs}

A key feature of non-unitary CFT's is that the lowest conformal dimension is not associated with the identity operator but with another operator that we denote by
$\mathit{o}$. Gannon argues that the degeneracy as well as multiplicity of the character $\chi_{\mathit{o}}$ are equal to 1\cite{Gannon:2003de}. One has
\[
\chi_{\mathit{o}} = q^{-\alpha_{\mathit{o}}} \big( 1+ O(q))\ ,
\]
where $\alpha_{\mathit{o}} = -\frac{c}{24}-h_{\mathit{o}}:=-\frac{c_\text{eff}}{24}$. In the MLDE approach, one na\"ively assigns the smallest exponent to the identity operator when it should be the operator $\mathit{o}$. This is called the unitarized version of the non-unitary operation.

In our determination of the exact $S$-matrix, this can potentially cause problems. There are two ways when this can be fixed or doesn't affect our computation. The first one, is to compute the fusion coefficients to determine which operator leads to non-negative coefficients and has a fusion matrix compatible with being the identity operator. The second one happens occasionally. Suppose there is a Galois symmetry that relates the identity operator to $\mathit{o}$ -- this is called the Galois shuffle by Gannon\cite{Gannon:2003de}. In such cases, $S_{00}=\pm S_{\mathit{o}\mathit{o}}$. Then, treating the operator $\mathit{o}$ as the identity will not mess up or exact determination of the $S$-matrix. In particular, this implies that
$\bar{S}_{ij} = \pm \frac{S_{ij}}{S_{\mathit{o}\mathit{o}}}:=\pm \bar{S}'_{ij}$. Thus, $\bar{S}'_{ij}\in \mathbb{Z}[\zeta_{\bar N}]$.

\section{Determining the S-matrix}

\subsection{The strategy}

The strategy that we employ is as follows:
\begin{enumerate}
\item Let $\mathbb{B}_0(w)$ and $\mathbb{B}_1(u)$ denote a basis of local solutions about $w=0$ and $u=(1-w)=0$ obtained by the Frobenius power series method. In these bases, the monodromy is diagonal. Let $M_0$ and $M_1$ represent the local monodromy matrices. Then, one has
\[
\mathbb{B}_0(w\, e^{2\pi i}) = M_0 \cdot \mathbb{B}_0(w)\quad,\quad \mathbb{B}_1(u\, e^{2\pi i}) = M_1 \cdot \mathbb{B}_1(u) \ .
\]
\item It is easy to see that $T=M_0$ as we usually work in the basis where the $T$ matrix is diagonal. However, $M_1$ is not the $S$ matrix. One has
\begin{equation}\label{SmatrixfromC}
S = C^{-1}\cdot M_1 \cdot C\ ,
\end{equation}
where $C$ is the \textit{connection} matrix that relates the two bases, $\mathbb{B}_0$ and $\mathbb{B}_1$. 
\item In particular, one has
\begin{equation}\label{Cdef}
\mathbb{B}_1(1-w_0) = C \cdot \mathbb{B}_0(w_0)\ ,
\end{equation}
where the matching has to be done at a regular point where the two Frobenius series both converge. In other words, we need $0<|w_0|<1$. We choose $w_0=\frac12+\epsilon$. By matching coefficients up to and including order $\epsilon^{n-1}$, we obtain $n^2$ linear equations that we use to numerically determine the entries in the connection matrix.
\item Once the connection matrix is obtained, we get the numerical estimate of the 
$S$-matrix.
\item We then convert the numerical estimate to an exact one as we explain next.
\end{enumerate}

\subsection{Obtaining exact expressions for the $S$-matrix}

Recall that we can determine $\bar N$ for the admissible solution. Next, we need a basis for $\mathbb{Q}(\zeta_{\bar{N}})$. 
The dimension of this given by $\varphi(N)$, where $\varphi$ is the Euler totient function. An integral basis for the ring of integers is given by $1,\zeta_N,\zeta_N^2,\ldots, \zeta_N^{\varphi(N)-1}$. Thus one has
\[
\bar{S}_{ij}=\sum_{m=0}^{\varphi(N)-1} a_{ij}(m)\ \zeta_N^m\ ,
\]
$a_{ij}(m)\in \mathbb{Z}$. This is a very useful expression as it lets us use our numerical estimate for $\bar{S}_{ij}$ to determine integers $a_{ij}(m)$  provide the closest to our numerical estimate to $\bar{S}_{ij}$. If it happens that $\bar{S}_{ij}$ is real, the number of terms reduce as we will see.

\subsection{A working example}

The first example of a non-hypergeometric MLDE occurs for three-character theories with Wronskian index $\ell=3$ i.e., $(n,\ell)=(3,3)$. The $S$-matrix for these theories were obtained using an involutive duality that related these theories to $(3,0)$ theories whose characters can be written in terms of the $_3F_2$ hypergeometric functions and their $S$-matrices are known. We will use one of those examples to re-derive the $S$-matrix using the method described earlier.

The $(3,3)$ MLDE is given as follows:
\begin{equation*}
\label{eq:33MDE}
    \left[ \mathcal{D}^3   +\frac12 \frac{E_4^2 }{E_6}\, \mathcal{D}^2 + \mu E_4\, \mathcal{D} + \nu E_6+\rho \frac{E_4^3 }{E_6}\right]\  f(\tau)= \, 0 \ ,
\end{equation*}
where 
\[
            \mu=e_2(\alpha) +\frac{1}{36} \quad,\quad  \nu=-e_3(\alpha)   \ ,    
\]
are fixed by the exponents. The accessory parameter, $\rho$, is left unfixed and has to be determined by the condition that the solutions be admissible.
The MLDE in the coordinates $w,u=(1-w)$ are as follows:
\begin{equation}
 \begin{aligned}
 \left[\partial_w^3+\left(\tfrac{3}{w}+\tfrac{1}{w-1}\right)\partial^2_w+\left(\tfrac{35/36+\mu}{w^2}+\tfrac{3/2-\mu}{w(w-1)}-\tfrac{1/4}{(w-1)^2}\right)\partial_w+\tfrac{\nu+\bar{\rho} w}{w^3(w-1)^2}\right] f(w)&=0~,\\
\left[\partial_u^3+\left(\tfrac{3}{u-1}+\tfrac{1}{u}\right)\partial^2_u+\left(\tfrac{35/36+\mu}{(u-1)^2}+\tfrac{3/2-\mu}{u(u-1)}-\tfrac{1/4}{u^2}\right)\partial_u+\tfrac{\nu+\bar\rho-\bar\rho u}{u^2(u-1)^3}\right]\ f(u)&=0~,
\end{aligned}
\end{equation}
where $\bar\rho =\frac{\rho}{1728}$.

Our working example has $c=\frac{48}7$, $h_1=\frac17$ and $h_2=\frac57$. From this we obtain the exponents
\[
\alpha_0=-\tfrac27\quad,\quad\alpha_1=-\tfrac17\quad,\quad \alpha_2=\tfrac37\ .
\]
From this we obtain $\mu=-(29/252)$, $\nu=-(6/343)$. The accessory parameter is $\bar\rho=(-989/24696)$. This fixes the MLDE of interest.
\begin{enumerate}
\item We first obtain the basis $\mathcal{B}_0$ which has the following $q$ expansion
\[
\mathcal{B}_0(w) =\begin{pmatrix}
q^{-2/7} \left(1 + O(q)\right) \\
q^{-1/7} \left(1 + O(q)\right) \\
q^{3/7} \left(55 + O(q)\right)
\end{pmatrix} = \begin{pmatrix}
(w/1728)^{-2/7} \left(1 + O(w)\right) \\
(w/1728)^{-1/7} \left(1 + O(w)\right) \\
(w/1728)^{3/7} \left(55 + O(w)\right)
\end{pmatrix}
\]
It is important to get the normalization of the solutions before working out the Frobenius power series about $w=0$. The above expansions fixes that ambiguity. 
Truncate the expansions at $q^{N_\text{max}+\alpha_i}$ for some large value of $N_\text{max}$. 
 \item The indicial equation about $u=0$ gives the indices $(0,\frac12,\frac32)$. Thus $M_1=\text{Diag}(1,-1,-1)$. Since the indices differ by an integer, we need to ensure that we get two linearly independent solutions. 
\[
\mathcal{B}_1(u) = \begin{pmatrix}
u^0\ \left( 1 + O(u)\right)\\
u^{1/2}\ \left( 1 + O(u)\right)\\
u^{3/2}\ \left( 1 + O(u)\right)
\end{pmatrix}
\]
\item Next we compute the connection matrix by considering the nine equations obtained from
\[
\mathbb{B}_1(\tfrac12 -\epsilon) - C \cdot \mathbb{B}_0(\tfrac12+\epsilon) = O(\epsilon^3)\ .
\]
We obtain
\[
C = \begin{pmatrix}
 0.0960231 & 0.0427724 & 0.0532802 \\
 1.42459 & -4.47652 & 1.02286 \\
 1.80508 & -6.30206 & 1.80122 \\
\end{pmatrix}
\]
\item Computing $S=C^{-1}\cdot M_1 \cdot C$, we obtain
\[
S=\begin{pmatrix}
0.327882 & 0.59149 & 0.7368 \\
 0.590957 & -0.736765 & 0.327904 \\
 0.7369 & 0.328244 & -0.591117 
 \end{pmatrix}\ ,
\]
which is the desired object. 
\item Since $\text{order}(\bar{T})=7$, we  expect the entries of $\bar{S}$ lie in the ring of integers $\mathbb{Z}(\zeta_7)$. 
Computing $\bar{S}$, we find
\[
\bar{S}= \begin{pmatrix} 
1. & 1.80397 & 2.24715 \\
 1.80235 & -2.24704 & 1.00007 \\
 2.24746 & 1.0011 & -1.80284 
 \end{pmatrix}  \ .
 \]
Since all the entries are real, a basis for them is given by
$(2\cos(2\pi x/7)$ for $x=1,2,3$. Writing the entries
as $\sum_{x=0}^2 a_x (\zeta_N^x+\zeta_N^{-x})$, we determine the integers $a_x$ by minimizing the following functions over integers in the range $[-4,4]$:
\[
F_{ij}(a_1,a_2,a_3):= \left|\bar{S}_{ij} - \sum_{x=0}^2  a_x (\zeta_N^x+\zeta_N^{-x}))\right|\ .\quad i,j=1,2,3\ .
\]
\begin{align*}
 \bar{S}= \begin{pmatrix}
 1 & \gamma_1 & \gamma_2\\
\gamma_1 & -\gamma_2 & 1 \\
 \gamma_2 & 1 &-\gamma_1
 \end{pmatrix} \ ,
\end{align*}
where $\gamma_1 =1+(\zeta_7 +\zeta_7^{-1})+(\zeta_7^2 +\zeta_7^{-2})$ and  $\gamma_2 = 1+(\zeta_7 +\zeta_7^{-1})$. This agrees with the $S$-matrix obtained using involutive duality\cite{Govindarajan:2025jlq}. We estimate the numerical accuracy of the connection matrix by defining the error matrix\footnote{We thank Emre Sert\"oz for emphasising the need to estimate errors.}
\begin{equation}
\text{ErrorM} = C\cdot S_{\text{exact}} \cdot C^{-1} - M_1\ ,
\end{equation}
and computing its Frobenius norm. The error goes down as we increase the number of terms in the Frobenius power series. In the figure, we show how the error changes with $N_\text{max}$.
\begin{figure}
\centering
\includegraphics[scale=0.7]{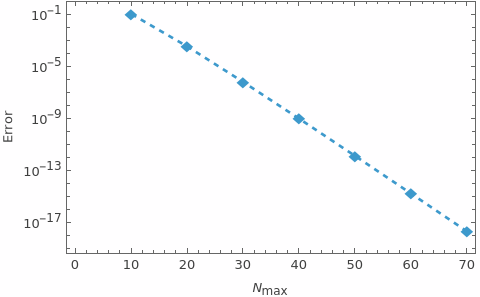}
\caption{We show how the error changes as we increase the number of terms in the Frobenius power series}
\end{figure}
\end{enumerate}

\subsubsection{Another $(3,3)$ example}

In computing the $S$-matrices for $(3,3)$ admissible solutions, two of the classes have $S_{00}=0$. We cannot define the normalized $S$-matrix in these cases. We need to use either the Coste-Gannon result that entries belong to $\mathbb{Q}(\zeta_N)$. 

We consider the theory with $c=5$ and $h_1=\frac1{16}$ and $h_{2}=\frac9{16}$.
We obtain the numerical S-matrix
\[
S_\text{num}=\begin{pmatrix}
0.\times 10^{-12} & 0.7071067812 &
   0.70710678119 \\
 0.70710678119 & -0.50000000000 & 0.50000000000 \\
 0.70710678119 & 0.5000000000 & -0.50000000000 
\end{pmatrix}
\]
The Coste-Gannon result implies that the element of the above matrix belong to $\mathbb{Q}(\zeta_{24})$. It turns out that it belongs to $\mathbb{Q}(\zeta_{8})$.
\[
S=\begin{pmatrix}
0 & \frac12 \sqrt{2} &
   \frac12 \sqrt{2} \\
 \frac12 \sqrt{2} & -\frac12 & \frac12 \\
\frac12 \sqrt{2} & \frac12 & -\frac12
\end{pmatrix}\ ,
\]
which agrees with the known answer\cite{Govindarajan:2025jlq}.
\begin{remark}
 It does not appear that there is a CFT associated with this $S$-matrix and hence the considerations of Coste-Gannon or de Boer-Goeree should not be applicable. We are na\"\i vely extending the validity of their results in this context. Surprisingly, it works rather well. Is there a more general derivation of Galois symmetry from the MLDE viewpoint without recourse to the Verlinde formula? 
\end{remark}

\subsection{Four-character examples}

In this sub-section, we consider two examples, one of $(4,0)$ type for which the S-matrices are known. The second example is of $(4,2)$ type that has not been studied so far and will be discussed in a forthcoming paper by two of the authors. 

The MLDE for $(4,\ell)$ theories is given by
\[
\left[ \mathcal{D}^4+ \mu_1 \frac{E_6}{E_4} \mathcal{D}^3 + \kappa E_4 \mathcal{D}^2  + \mu E_6 \mathcal{D} +(\nu E_4^2+\rho \frac{\Delta}{E_4})\right]  f(\tau)\,= \, 0\ ,
\]
where
\begin{align*}
&\mu_1(\ell)=\frac{\ell}{6}\ ,\ \kappa(\ell)=e_2(\alpha)-\frac{11-3\ell}{36} \ ,\ 
\\
&\mu(\ell)=-e_3(\alpha)+\frac{1}{6}e_2(\alpha)-\frac{5-\ell}{216}\ ,\ \nu=e_4(\alpha),
 \end{align*}
 For $\ell=0$, the accessory parameter $\rho$ vanishes while it has to be fixed on a case by case basis for $\ell=2$.

\subsection*{A $(4,0)$ example}

The theory has $c_\text{eff}=\frac65$, (shifted) weights $(0,\frac15,\frac25,\frac35)$. It is obtained as a tensor product $LY_1^{\otimes3}$ and has four characters and eight primaries.  This is a non-unitary theory but we are giving the unitarized version of the central charge and weights. The identity operator is assigned the weight $\frac35$. Our calculations use the unitarized version. The lowest dimension operator is related to the identity operator by Galois symmetry and so $\bar{S}\in\mathbb{Z}[\zeta_5]$.

The $q$-series has the following leading expansion:
\[
\mathcal{B}_0(w) =\begin{pmatrix}
q^{-1/20} \left(1 + O(q)\right) \\
q^{3/20} \left(1 + O(q)\right) \\
q^{7/20} \left(1 + O(q)\right)\\
q^{11/20} \left(1 + O(q)\right)
\end{pmatrix}
\]
Repeating the procedure outlined, we obtain the normalized $S$-matrix from which we obtain the exact S-matrix.
\[
\bar{S}=\begin{pmatrix}
1.00000 & 1.85410 & 1.14590 & 0.236068 \\
 0.618034 & -0.236068 & -1.00000 & -0.381966 \\
 0.381966 & -1.0000 & 0.236068 & 0.618034 \\
 0.236068 & -1.14590 & 1.85410 & -1.00000 
\end{pmatrix} =
\begin{pmatrix}
1 & 3 \gamma_1 & 3\gamma_2 & \gamma_3 \\
\gamma_1 & -\gamma_3 & -1 & -\gamma_2 \\
\gamma_2 & -1 & \gamma_3 & \gamma_1 \\
\gamma_3 & 3 \gamma_2 &   3\gamma_1 & -1 
\end{pmatrix}
\]
where $\gamma_1=\zeta_5 + \zeta_5^{-1}$, $\gamma_2=1-\gamma_1$, $\gamma_3=-1+2\gamma_1$. This is consistent with $\bar{N}=5$. It is easy to see that the multiplicities of the characters are $(1,3,3,1)$ -- thus this is an eight primary CFT.

\subsection*{A $(4,2)$ example}

The admissible has $c=\frac83$, weights $(0,\frac19,\frac13,\frac23)$. The $q$-series has the following leading expansion:
\[
\mathcal{B}_0(w) =\begin{pmatrix}
q^{-1/9} \left(1 + 3 q + O(q^2)\right) \\
q^{0} \left(1 + O(q)\right) \\
q^{2/9} \left(5 + O(q)\right)\\
q^{5/9} \left(7 + O(q)\right)
\end{pmatrix}
\]
We give an extra term in the identity character as there is one unfixed accessory parameter. Here $\rho=-280/81$. The degeneracies for the last two characters are $5$ and $7$.

The normalized S-matrix that we obtain is as follows:
\[
\bar{S} = \begin{pmatrix}
1.00000 & 8.63816 & 4.41147 & 5.41147 \\
 2.87939 & 2.6\times 10^{-12} & 2.87939 & -2.87939 \\
 4.41147 & 8.63816 & -5.41147 & -1.00000 \\
 5.41147 & -8.63816 & -1.00000 & 4.41147 \\
\end{pmatrix}=\begin{pmatrix}
1 & 3 \gamma_1 & \gamma_2  & 1+ \gamma_2\\
\gamma_1 & 0 & \gamma_1 & -\gamma_1 \\
\gamma_2 & 3\gamma_1 & -1-\gamma_2 & -1\\
1+\gamma_2 & -3 \gamma_1 & -1 & \gamma_2
\end{pmatrix}\ ,
\]
where $\gamma_1 = 1+ (\zeta_9 +\zeta_9^{-1})+ (\zeta_9^2 +\zeta_9^{-2})$, $\gamma_2=\gamma_1+ (\zeta_9 +\zeta_9^{-1})$. It is easy to see that the multiplicities of the characters are $(1,3,1,1)$ -- thus this is an six primary CFT. The resolved S-matrix appears in the list of  six-primary MTCs\cite{Ng:2023}.

\subsection{Five-character examples}

We consider the determination of the S-matrix for admissible solutions for five-character MLDEs. 

\subsection*{A pair of five character examples}

We consider a pair of five character unitary theories, one of $(5,0)$ type and the other of type $(5,2)$ that are GHM duals of each other. The MLDE for $(5,\ell)$ $(\ell=0,2)$  theories here is given by
\[\left[ \mathcal{D}^5 + \mu_1\frac{E_6}{E_4}\mathcal{D}^4+\kappa E_4 \mathcal{D}^3+ \mu_3 E_6 \mathcal{D}^2  + \left(\mu_4 E_4^2 +\rho \frac{\Delta}{E_4}\right)\mathcal{D} +\nu E_4E_6\right]  f(\tau)\,= \, 0~.\]
where
\begin{align*}
    \mu_1(\ell)&=\frac{\ell}{6},\qquad   \kappa(\ell)=e_2(\alpha)-\frac{35-6\ell}{36},\qquad
\mu_3(\ell)=-e_3(\alpha)+\frac{1}{2}e_2(\alpha)-\frac{55-7\ell}{216},\\
 \mu_4(\ell)&=e_4(\alpha)-\frac{1}{6}e_3(\alpha)+\frac{1}{36}e_2(\alpha)-\frac{9-\ell}{1296}, \qquad \nu=-e_5(\alpha).
\end{align*}
Further, the accessory parameter, $\rho$, vanishes for $\ell=0$. 

The two theories under consideration have the following properties.
\begin{enumerate}
        \item The Hecke operator $T_{31}$ acting on $LY_4$-RCFT gives rise to a $(5,0)$ theory with $c=\frac{248}{11}$ and $h=(0,\frac{20}{11},\frac{16}{11},\frac{21}{11},\frac{13}{11})$, and the  $q$-series has the following leading expansion: 
\[
\mathcal{B}_0^{(1)}(w) =
\begin{pmatrix}
q^{-31/33}\left(1 +O(q)\right) \\
\mathbf{2}\,q^{29/33}\left(15314 + O(q)\right) \\
q^{17/33}\left(4123 + O(q)\right) \\
\mathbf{31}\,q^{32/33}\left(1976 + O(q)\right) \\
\mathbf{31}\,q^{8/33}\left(8 + O(q)\right)
\end{pmatrix}.
\]
corresponds to the characters of RCFT associated to the Thompson group. The integrality constraints in MLDE approach do not fix the degeneracies that are indicated above in bold face. However, the Hecke operator method fixes it to the values given above. The $S$-matrix is in the $LY_4$-class.
 \item The Hecke $T_2$ action on $LY_4$ gives $(5,2)$ theory, which is $c_{\mathcal{H}}=24$ GHM dual of the above theory. The theory has $c=\frac{16}{11}$ and $h=(0,\frac2{11},\frac{6}{11},\frac{1}{11},\frac{9}{11})$, and corresponds to the 3C conjugacy class of the Monster group. The $q$-series has the following leading expansion: 
\[
\mathcal{B}_0^{(2)}(w) =
\begin{pmatrix}
q^{-2/33}\left(1 + 0\,q + O(q^2)\right) \\
q^{4/33}\left(1 + O(q)\right) \\
q^{16/33}\left(3 + O(q)\right) \\
q^{1/33}\left(2 + O(q)\right) \\
q^{25/33}\left(2 + O(q)\right)
\end{pmatrix}.
\]
The accessory parameter is found to be $\rho = -\frac{1480}{121}.$ This can be determined in two ways: first, the Hecke operator gives the coefficient of $q$ in the identity character, zero in this case. The second way is via the integrality constraint in the MLDE approach. The $S$-matrix is in the $LY_4$-class. It is easy to check that 
$$(\mathcal{B}_0^{(1)})^T\cdot \mathcal{B}_0^{(2)} = J(\tau)-744\ .
$$
This shows that the two theories are GHM duals\cite{Gaberdiel:2016zke}.
\end{enumerate}

As expected our numerical computation for these two theories give the same $S$-matrix. The normalized $S$-matrix that we obtain from one of these two MLDEs is
\begin{align*}
   \bar{S} &=
\begin{pmatrix}
 1.00000 & 2.68251 & 3.51334 & 3.22871 & 1.91899 \\
 2.68251 & 1.91899 & -3.22871 & -1.00000 & 3.51334 \\
 3.51334 & -3.22871 & 2.68251 & -1.91899 & 1.00000 \\
 3.22871 & -1.00000 & -1.91899 & 3.51334 & -2.68251 \\
 1.91899 & 3.51334 & 1.00000 & -2.68251 & -3.22871 \\
\end{pmatrix}\\
&=
\begin{pmatrix}
1 & \gamma_1 & \gamma_2 &  \gamma_3 & \gamma_4 \\
\gamma_1 & \gamma_4 & -\gamma_3 & -1 & \gamma_2 \\
\gamma_2 & -\gamma_3 & \gamma_1  &  -\gamma_4  & 1 \\
\gamma_3 & -1 &  -\gamma_4 & \gamma_2 & -\gamma_1\\
\gamma_4 & \gamma_2 & 1 & -\gamma_1 &-\gamma_3
\end{pmatrix}, 
\end{align*}
where  $\gamma_1 = 1 + (\zeta_{11} + \zeta_{11}^{-1})$,
$\gamma_2 = \gamma_1+(\zeta_{11}^2 + \zeta_{11}^{-2})$,
$\gamma_3 = \gamma_2+(\zeta_{11}^3 + \zeta_{11}^{-3})$ and
$\gamma_4 = \gamma_3+(\zeta_{11}^4 + \zeta_{11}^{-4})$. This is consistent with the order $\bar{T}$ being $11$ in both cases.

\section{Conclusion}

The main focus of this paper has been to compute exact S-matrices within the setting of the holomorphic modular bootstrap. This has been illustrated with a selected set of examples arising as admissible solutions of MLDE. A larger set of examples will be provided as an ancillary Mathematica file with the arXiv submission of this paper.

The exact $S$-matrices were obtained by appealing to integrality results from Galois symmetry in RCFTs. However, the set of admissible solutions is larger than the set of RCFTs. Do such integrality results hold for these solutions as well? We have implicitly assumed that it must be true for all admissible solutions. It is of interest to prove the Galois symmetry using the theory of differential Galois theory to the case of MLDEs\cite{Beukers:1995,Magid:1994,Pillay:1998,Put:2012}. Is there a version of the Verlinde formula in this context?

The $F$ and $R$ symbols that extend the modular data of an RCFT are additional objects that need to be computed. In the context of MTC's, these have been completed computed for those MTC's that are associated with four or fewer primaries (simple objects)\cite{Rowell:2009}. In the context of RCFT's, the $F$-symbols correspond to a change of basis for conformal blocks associated with four-point functions on the sphere\cite{Cheng:2020srs}. 
The methods used in this paper might be of some use in this context. 
The R-symbols are cyclotomic integers in $\mathbb{Q}(\zeta_{2\bar N})$ under a suitable gauge choice. However, the F-symbols remain unknown but the expectation is that they are in 
$\mathbb{Q}(\zeta_{2\bar N})$\cite[see Theorem 3.1]{Ng:2024}.\footnote{We thank Prof. Ng for an email correspondence explaining this result.} Is there an integrality statement in these cases?


\subsection*{Acknowledgments} We thank Emre Sert\"oz and Abhiram Kidambi for useful discussions. Abhiram helped us better understand the integrality result for entries of the normalized $S$-matrices. We have included his explanation as an appendix so that others can benefit as we did. We also thank Richard Ng for useful email correspondences regarding his work. This work arose from  one of us (SG) attending the ICTP/ Winter School on Number Theory and Physics(smr 4162) held during November 2025. We thank the organisers for the kind invitation. AS was also supported by a CSIR Senior Research Fellowship.

\appendix

\section{Some definitions}

The $n$-the root of unity is $\zeta_n:=\exp(2\pi i/n)$.
We use the notation $e_n(\alpha)$ to represent symmetric polynomials of the exponents $(\alpha_0,\ldots,\alpha_{n-1})$. Thus, one has
\begin{align*}
e_1(\alpha) &= (\alpha_0+\cdots +\alpha_{n-1})\\
e_2(\alpha) &= (\alpha_0 \alpha_1 + \text{permutations})\\
e_3(\alpha) &= (\alpha_0\alpha_1 \alpha_2 + \text{permutations})
\end{align*}

\subsection{Modular background}

The formulae for the Eisenstein series ($E_2$, $E_4,$ and $E_6$), the cusp form of weight $12$ ($\Delta$), and Klein-$J$ invariant are given. Let  $q=\exp(2\pi i \tau)$. 
\begin{align*}
E_2(\tau) &= 1 - 24\sum_{n=1}^{\infty} \frac{n\, q^n}{(1-q^n)} = 1-24q-72q^2-96q^3  +\cdots \\
E_4(\tau) &= 1 + 240\sum_{n=1}^{\infty} \frac{n^3\, q^n}{(1-q^n)} = 1 + 240 q + 2160 q^2 + 6720 q^3  + \cdots \\
E_6(\tau) &= 1 - 504\sum_{n=1}^{\infty} \frac{n^5\, q^n}{(1-q^n)} = 1 - 504 q - 16632 q^2 - 122976 q^3 +\cdots \\
\Delta(\tau) &= \eta^{24} =  \frac{E_4^3 - E_6^2}{1728} = q - 24 q^2 + 252 q^3 - 1472 q^4 + 4830 q^5 + \cdots \\
J(\tau) &= \frac{E_4^3}{\Delta} = \frac{1}{q} + 744 + 196884 q + 21493760 q^2 +864299970 q^3+\cdots
\end{align*}
 The covariant derivative, acting on modular forms of weight $w$ is,  
\[
\mathcal{D}_w := q\frac{d}{dq} - \frac{w}{12} E_2.
\]
and the higher order derivatives are defined as
\[
\mathcal{D}_w^n :=  \mathcal{D}_{w+2n-2}\circ \mathcal{D}_{w+2n-4}\circ\ldots \circ\mathcal{D}_w.
\]

\subsection{The LY minimal models}

Theories in the $LY_1$ class (defined by Duan et al.\cite{Duan:2022ltz}) arise from the non-unitary minimal model, $M(5,2)$ which is  associated with Lee-Yang edge singularity. It has
\[
c=-\frac{22}{5}\quad,\quad h=(0.-\frac{1}{5})\ .
\]
Defining $c_\text{eff} := (c -24 h_\text{min})$ and $h_\text{eff}:=(h-h_\text{min})$, we obtain a `unitarized' version given by
\[
c_\text{eff}= \frac{2}{5}\quad,\quad h_\text{eff}=\left(0,\frac{1}{5}\right)\ .
\]
Theories in the $LY_1$ class appear as $(2,4)$  in the list of MTCs.

Theories in the $LY_4$ class (defined by Duan et al.\cite{Duan:2022ltz}) arise from the non-unitary minimal model $M(11,2)$. It has
\[
c=-\frac{232}{11}\quad,\quad h=\left(0,-\frac{4}{11},-\frac{7}{11},-\frac{9}{11},-\frac{10}{11}\right)\ .
\]
The `unitarized' version is given by
\[
c_\text{eff}=\frac{8}{11}\quad,\quad h_\text{eff}=\left( 
0,\frac{1}{11},\frac{3}{11},\frac{6}{11},\frac{10}{11}  \right)\ .
\]
Duan et al. construct a family of five-character theories by the action of various Hecke operators on the characters of the non-unitary minimal model $M(11,2)$.


\section{Proof of cyclotomic integrality of S-matrix entries}\label{Kidambi}

\begin{center}
by \\[5pt]
Abhiram Kidambi\footnote{This proof was explained to the authors of the manuscript (SG, AJ, AS) by Abhiram Kidambi.  While AK felt a brief acknowledgment would have sufficed, he appreciates the authors' thoughtfulness and scientific ethos on featuring his explanation as a standalone appendix. AK is funded by the ERC (MaScAmp, 101167287).}
\end{center}
%

This appendix contains a short proof of cyclotomic integrality of S-matrix entries.
\noindent Let $\mathcal{C}$ be a modular tensor category (MTC) with a finite set of isomorphism classes of simple objects $\{X_0, \dots, X_{n-1}\}$, where $X_0$ is the monoidal unit. The \emph{normalized} modular data $(\bar{S}, \bar{T})$ provides a projective representation of $SL(2, \mathbb{Z})$.  
\begin{remark}
\label{rmk1}
The matrices  $\bar S$ and $\bar{T}$ are normalized meaning that ${\bar S}_{00}={\bar T}_{00} = 1$. Additionally, ${\bar S}$ is symmetric. The ${\bar S}$-matrices in the main text have size equal to the number of characters while in the MTC setting, the rank is equal the number of primaries. 
\end{remark}
    
    \begin{proposition}
    \label{mainprop}
Let ${\bar N}$ (taken to be a finite positive integer) be the order of ${\bar T}$.
For any $i, j \in \mathbb{Z}\cap [0,n-1]$, $\bar{S}_{ij}\in \mathbb{Z}[\zeta_{\bar N}]$, where $\mathbb{Z}[\zeta_{\bar N}]$ is the integer ring of the  ${\bar N}$ dimensional cyclotomic extension over $\mathbb{Q}$ i.e, $\mathbb{Q}(\zeta_{\bar N})$. 
\end{proposition}

We want to provide a proof of  \Cref{mainprop}.\footnote{Note that this statement appears in \cite[Proposition II.1]{Ng:2023}, and this appendix is aimed at providing a compact proof of this statement. } This proof given here is essentially a distillation of the relevant results found in \cite{Ng:2007}. We begin by building up with the relevant definitions. 

\begin{definition}[Fusion ring]
The fusion ring $K(\mathcal{C})$ is the $\mathbb{Z}$-algebra generated by simple objects with multiplication defined by
\[
X_i \otimes X_j \cong \bigoplus_{k=0}^{n-1} N_{ij}^k\, X_k,
\]
where $N_{ij}^k = \dim \mathrm{Hom}(X_k, X_i \otimes X_j) \in \mathbb{Z}_{\ge 0}$.
\end{definition}

\begin{fact}
For each $i$, let $N_i$ be the fusion matrix such that $(N_i)_{jk} = N_{ij}^k$. Then $N_i \in \mathrm{Mat}_{n\times n}(\mathbb{Z})$.
\end{fact}

\begin{lemma}[de Boer--Goeree, \cite{DeBoer:1990em}]
\label{lemma:vf}
The ${\bar S}$-matrix simultaneously diagonalizes the fusion matrices. In particular, the eigenvalues of $N_i$ are given by
\[
\lambda_i^{(m)} = \frac{S_{im}}{S_{0m}}=\frac{\bar{S}_{im}}{\bar{S}_{0m}}, \qquad m = 0, \dots, n-1.
\]
and $\lambda_i^{(m)}$ are algebraic integers.

\end{lemma}

\begin{proof}[Proof of \Cref{lemma:vf}]
Since $N_i$ has integer entries, its characteristic polynomial is monic with coefficients in $\mathbb{Z}$. Meaning every eigenvalue of $N_i$ 
is an algebraic integer. By the Verlinde formula, these eigenvalues are given by $\lambda_i^{(m)}$.
So $\lambda_i^{(m)}$ $\forall \ i,m \in \mathbb{Z} \cap [0,n-1]$ is an algebraic integer.
\end{proof}
\begin{remark}
%
  de Boer-Goeree in fact show that this is an algebraic integer in some abelian extension of $\mathbb{Q}$\cite{DeBoer:1990em}. The Kronecker--Weber theorem states that every finite abelian extension of $\mathbb{Q}$ is contained in a cyclotomic field. Since we will use the stronger theorem of Ng-Schauenberg stated below which establishes the precise cyclotomic extension, it suffices to show for now that the $\lambda_i^{(m)}$ are algebraic integers. 
\end{remark}

\begin{theorem}[Ng--Schauenberg, \cite{Ng:2007}]
\label{thm:ns}
Let $\rho: SL(2, \mathbb{Z}) \to GL(n, \mathbb{C})$ be the representation associated with $\mathcal{C}$, and let $\bar{N}$ be the order of the matrix $\bar{T}$. Then
\begin{enumerate}
    \item The kernel of $\rho$ is a congruence subgroup of level $\bar{N}$.
    \item The entries of the matrices $\bar{S}$ and $\bar{T}$ lie in the cyclotomic field $\mathbb{Q}(\zeta_{\bar{N}})$.
\end{enumerate}
\end{theorem}
\begin{proof}
    We refer the reader to Theorem 6.8, and Propositions 5.7 and 5.8 of \cite{Ng:2007}.
\end{proof}

\begin{proposition}
\label{mainthm2}
$
\lambda_i^{(m)} \in \mathbb{Z}[\zeta_{\bar N}]$.   
\end{proposition}

\begin{proof}[Proof of \Cref{mainthm2}]
Combining \Cref{lemma:vf} and \Cref{thm:ns}, it follows that
\[
\lambda_i^{(m)}  \in \mathbb{Z}[\zeta_{\bar N}].
\]
\end{proof}
 We can now put together the proof of \Cref{mainprop}.
\begin{proof}[Proof of \Cref{mainprop}]
\label{mainproof}
$\bar{S}_{ij} = \lambda_i^{(j)}\, \lambda_{j}^{(0)}$. 
Since $\mathbb{Z}[\zeta_{\bar N}]$ is a ring (closed under multiplication), we see immediately that
\[
\bar{S}_{ij} \in \mathbb{Z}[\zeta_{\bar N}]\ .
\]
\end{proof}

\bibliographystyle{jhep}
\bibliography{master}  
\end{document}